\documentclass[a4paper,11pt]{amsart}
\usepackage[utf8]{inputenc}
\usepackage{amsaddr}
\usepackage{calrsfs}
\usepackage{graphicx,color}

\numberwithin{equation}{section}

\theoremstyle{definition}
\newtheorem{dfn}{Definition}[section]

\theoremstyle{plain}
\newtheorem{thm}{Theorem}[section]

\newtheorem{cor}{Corollary}[section]
\newtheorem{lmm}{Lemma}[section]
\theoremstyle{definition}
\newtheorem{rem}{Remark}[section]

\newcommand{\R}{\mathbb{R}}
\newcommand{\E}{\mathbb{E}}

\newcommand{\F}{\mathcal{F}}

\newcommand{\e}{\mathrm{e}}

\newcommand{\1}{\mathbf{1}}
\renewcommand{\d}{\mathrm{d}}

\begin{document}
\title[Hedging with fractional Brownian motion]
{Hedging in fractional Black--Scholes Model with Transaction Costs}

\date{\today}

\author[Shokrollahi]{Foad Shokrollahi}
\address{Department of Mathematics and Statistics, University of Vaasa, P.O. Box 700, FIN-65101 Vaasa, FINLAND}
\email{foad.shokrollahi@uva.fi}

\author[Sottinen]{Tommi Sottinen}
\address{Department of Mathematics and Statistics, University of Vaasa, P.O. Box 700, FIN-65101 Vaasa, FINLAND}
\email{tommi.sottinen@iki.fi}

\begin{abstract}
We consider conditional-mean hedging in a fractional Black--Scholes pricing model in the presence of proportional transaction costs.
We develop an explicit formula for the conditional-mean hedging portfolio in terms of the recently discovered explicit conditional law of the fractional Brownian motion.
\end{abstract}

\thanks{F. Shokrollahi was funded by the graduate school of the University of Vaasa.
\\
We thank the referee for valuable comments.}



\maketitle

\section{Introduction}

%

We consider discrete hedging in fractional Black--Scholes models where the asset price is driven by a long-range dependent fractional Brownian motion.  For a convex or concave European vanilla type option, we construct the so-called conditional-mean hedge.  This means that at each trading time the value of the conditional mean of the discrete hedging strategy coincides with the frictionless price.  By frictionless we mean the continuous trading hedging price without transaction costs. The key ingredient in constructing the conditional mean hedging strategy is the recent representation for the regular conditional law of the fractional Brownian motion given in \cite{Sottinen-Viitasaari-2017b}.  Let us note that there are arbitrage strategies with continuous trading without transaction costs, but not with discrete trading strategies, even in the absence of trading costs.  For details of the use of fractional Brownian motion in finance and discussion on arbitrage see \cite{Bender-Sottinen-Valkeila-2011}.

For the classical Black--Scholes model driven by the Brownian motion, the study of hedging under transaction costs goes back to Leland \cite{Leland-1985}. See also
Denis and Kabanov \cite{Denis-Kabanov-2010} and Kabanov and Safarian \cite{Kabanov-Safarian-2009} for a mathematically rigorous treatment.
For the fractional Black--Scholes model driven by the long-range dependent fractional Brownian motion, the study of hedging under transaction costs was studied in Azmoodeh \cite{Azmoodeh-2013}.
In the series of articles \cite{Shokrollahi-Kilicman-Magdziarz-2016,Wang-2010a,Wang-2010b,Wang-Yan-Tang-Zhu-2010,Wang-Zhu-Tang-Yan-2010} the discrete hedging in the fractional Black--Scholes model was studied by using the economically dubious Wick--It\^o--Skorohod interpretation of the self-financing condition.  Actually, with the economically solid forward-type pathwise interpretation of the self-financing condition, these hedging strategies are valid, not for the geometric fractional Brownian motion, but for a geometric Gaussian process where the driving noise is a Gaussian martingale with the same variance function as the corresponding fractional Brownian motion would have, see \cite{Gapeev-Sottinen-Valkeila-2011,Shokrollahi-Kilicman-2014,Shokrollahi-Kilicman-2015,Shokrollahi-Kilicman-Ibrahim-Ismail-2015}.

\section{Preliminaries}

We are interested in pricing of European vanilla options $f(S_T)$ of a single underlying asset $S=(S_t)_{t\in[0,T]}$, where $T>0$ is a fixed time of maturity of the option.

We consider the discounted fractional Black--Scholes pricing model where the ``riskless'' investment, or the bond, is taken as the num\'eraire and the risky asset $S=(S_t)_{t\in[0,T]}$ is given by the dynamics
\begin{equation}\label{eq:sde-S}
\frac{\d S_t}{S_t} = \mu\, \d t + \sigma\, \d B_t,
\end{equation}
where $B$ is the fractional Brownian motion with Hurst index $H\in(\frac12,1)$.
Recall that, qualitatively, the fractional Brownian motion is the (up to a multiplicative constant) unique Gaussian process with stationary increments and self-similarity index $H$.  Quantitatively, the fractional Brownian motion is defined by its covariance function
$$
r(t,s) = \frac12\left[t^{2H}+s^{2H}-|t-s|^{2H}\right].
$$
Since the fractional Brownian motion with index $H\in(\frac12,1)$ has zero quadratic variation, the classical change-of-variables rule applies.  Consequently, the pathwise solution to the stochastic differential equation \eqref{eq:sde-S} is
\begin{equation}\label{eq:S-from-B}
S_t = S_0 \e^{\mu t + \sigma B_t}.
\end{equation}
Also, it follows from the classical change-of-variables rule that
\begin{equation}\label{eq:hedging-rule}
f(S_T) = f(S_0) + \int_0^T f'(S_t)\, \d S_t,
\end{equation}
where $f$ is a convex or concave function and $f'$ is its left-derivative.  We refer to Azmoodeh et al. \cite{Azmoodeh-Mishura-Valkeila-2009} for details.
The economic interpretation of \eqref{eq:hedging-rule} is that under continuous trading and no transaction costs, the replication price of a European vanilla option $f(S_T)$ is $f(S_0)$ and the replicating strategy is given by $\pi_t = f'(S_t)$, where $\pi_t$ denotes the number of the shares of the risky asset $S$ held by the investor at time $t$. Furthermore, we note that the value $V^{\pi}$ of the hedging strategy $\pi=f'(S_\cdot)$ at time $t$ is
\begin{eqnarray*}
V^\pi_t &=& V^\pi_0 + \int_0^t \pi_u\, \d S_u \\
&=&
f(S_0) + \int_0^t f'(S_u)\, \d S_u \\
&=&
f(S_t).
\end{eqnarray*}
Indeed, the first equality is simply the self-financing condition and the rest follows immediately from \eqref{eq:hedging-rule}. Note that this is very different from the value in the classical Black--Scholes model, where the value is determined by the Black--Scholes partial differential equation, which in turn comes to the It\^o's change-of-variables rule.

We assume that the trading only takes place at fixed preset time points $0=t_0<t_1<\cdots<t_N=T$. We denote by $\pi^N$ the discrete trading strategy
$$
\pi^N_t = \pi^N_0\1_{\{0\}}(t) + \sum_{i=1}^N \pi_{t_{i-1}}^N\1_{(t_{i-1},t_{i}]}(t).
$$
The value of the strategy $\pi^N$ is given by
\begin{equation}\label{eq:kappa}
V^{\pi^N,k}_t = V^{\pi^N,k}_0 + \int_0^t \pi^N_u \, \d S_u -
\int_0^t kS_{u} |\d\pi^N_{u}|,
\end{equation}
where $k\in[0,1)$ is the proportional transaction cost.

Under transaction costs perfect hedging is not possible.  In this case, it is natural to try to hedge on average in the sense of the following definition:

\begin{dfn}[Conditional-Mean Hedge]\label{dfn:cm-hedge}
Let $f(S_T)$ be a European vanilla type option with convex or concave payoff function $f$.  Let $\pi$ be its Markovian replicating strategy: $\pi_t=f'(S_t)$.  We call the discrete-time strategy $\pi^N$ a \emph{conditional-mean hedge}, if for all trading times $t_i$,
\begin{equation}\label{eq:cm-hedge}
\E\left[ V^{\pi^N, k}_{t_{i+1}}\,|\, \F_{t_i}\right] = \E\left[ V^{\pi}_{t_{i+1}}\,|\, \F_{t_i}\right].
\end{equation}
Here $\F_{t_i}$ is the information generated by the asset price process $S$ up to time $t_i$.
\end{dfn}

\begin{rem}[Conditional-Mean Hedge as Tracking Condition]
Criterion \eqref{eq:cm-hedge} is actually a tracking requirement.  We do not only require that the conditional means agree on the last trading time before the maturity, but also on all trading times.  In this sense the criterion has an ``American'' flavor in it.  From a purely ``European'' hedging point of view, one can simply remove all but the first and the last trading times.
\end{rem}

\begin{rem}[Arbitrage and Uniqueness of Conditional-Mean Hedge]
The conditional-mean hedging strategy $\pi^N$ depends on the continuous-time hedging strategy $\pi$.
Since there is strong arbitrage in the fractional Black--Scholes model (zero can be perfectly replicated with negative initial wealth), the replicating strategy $\pi$ is not unique. However, the strong arbitrage strategies are very complicated.  Indeed, it follows directly from the change-of-variables formula that in the class of Markovian strategies $\pi_t = g(t,S_t)$, the choice $\pi_t=f'(S_t)$ is the unique replicating strategy for the claim $f(S_T)$.
\end{rem}

We stress that the expectation in \eqref{eq:cm-hedge} is with respect to the true probability measure; not under any equivalent martingale measure.  Indeed, equivalent martingale measures do not exist in the fractional Black--Scholes model.

To find the solution to \eqref{eq:cm-hedge} one must be able to calculate the conditional expectations involved.  This can be done by using \cite[Theorem 3.1]{Sottinen-Viitasaari-2017b}, a version of which we state below as Lemma \ref{lmm:cond-fbm} for the readers' convenience.

\begin{lmm}[Conditional Fractional Brownian Motion]\label{lmm:cond-fbm}
The fractional Brownian motion $B$ with index $H\in(\frac12,1)$ conditioned on its own past $\F^B_u$ is the Gaussian process $B(u)=B|\F^B_u$ with $\F^B_u$-measurable mean
\begin{equation*}
\hat B_t(u) = B_u - \int_0^u \Psi(t,s|u) \, \d B_s,
\end{equation*}
where
$$
\Psi(t,s|u) =
- \frac{\sin(\pi(H-\frac{1}{2}))}{\pi} s^{\frac{1}{2}-H}(u-s)^{\frac{1}{2}-H}
\int_u^t \frac{z^{H-\frac{1}{2}}(z-u)^{H-\frac{1}{2}}}{z-s}\, \d z,
$$
and deterministic covariance function
\begin{equation*}
\hat r(t,s|u)
= r(t,s) - \int_0^u k(t,v)k(s,v)\d v,
\end{equation*}
where
$$
k(t,s) =
\left(H-\frac12\right)
\sqrt{\frac{2H\Gamma\left(\frac32-H\right)}{\Gamma\left(H-\frac12\right)\Gamma\left(2-2H\right)}}\,\,
s^{\frac12-H}\int_s^t z^{H-\frac12}(z-s)^{H-\frac23}\, \d z;
$$
$\Gamma$ is the Euler's gamma function.
\end{lmm}

\begin{rem}[Conditional Asset Process]\label{rem:cond-asset}
By \eqref{eq:S-from-B} we have the equality of filtrations: $\F_t^B = \F^S_t = \F_t$, for $t\in [0,T]$. Consequently, the conditional process $S(u) = S|\F_u$ is, informally, given by
\begin{eqnarray*}
S_t(u) &=& S_0\e^{\mu t + \sigma B_t(u)} \\
&=&
S_u\e^{\mu (t-u) + \sigma\left(B_t(u)-B_u\right)}.
\end{eqnarray*}
More formally, this means, in particular, that for $t>u$,
\begin{eqnarray*}
\E\left[f(S_t)\,\big|\, \F_u^S\right]
&=&
\E\left[f\left(S_0\e^{\mu t + \sigma B_t}\right) \, \big|\, \F_u^B\right] \\
&=&
\int_{-\infty}^\infty f\left(S_0\e^{\mu t + \sigma \hat B_t(u)+\sigma\sqrt{\hat r(t|u)}\, z}\right)\, \phi(z) \d z \\
&=&
\int_{-\infty}^\infty f\left(S_u\e^{\mu(t-u) + \sigma\left(\hat B_t(u)-B_u\right)+\sigma\sqrt{\hat r(t|u)}\, z}\right)\, \phi(z) \d z,
\end{eqnarray*}
where we have denoted
$$
\hat r(t|u) = \hat r(t,t|u),
$$
and $\phi$ is the standard normal density function.
\end{rem}

\section{Conditional-Mean Hedging Strategies}

Denote
$$
\Delta\hat B_{t_{i+1}}(t_i)
= \hat B_{t_{i+1}}(t_i) - B_{t_i}.
$$
In Theorem \ref{thm:dhedging} we will calculate the conditional-mean hedging strategy in terms of the following conditional gains:
\begin{eqnarray*}
\Delta\hat S_{t_{i+1}}(t_i)
&=&
\hat S_{t_{i+1}}(t_i) - S_{t_i} \\
&=&
\E\left[S_{t_{i+1}}|\F_{t_i}\right] - S_{t_i}, \\
\Delta\hat V^\pi_{t_{i+1}}(t_i)
&=&
\hat V^\pi_{t_{i+1}}(t_i) - V^\pi_{t_i} \\
&=&
\E\left[V^\pi_{t_{i+1}}|\F_{t_i}\right] - V^\pi_{t_i}, \\
\Delta\hat V^{\pi^N,k}_{t_{i+1}}(t_i)
&=&
\hat V^{\pi^N,k}_{t_{i+1}}(t_i) - V^\pi_{t_i} \\
&=&
\E\left[V^{\pi^N,k}_{t_{i+1}}|\F_{t_i}\right] - V^{\pi^N,k}_{t_i}.
\end{eqnarray*}
Lemma \ref{lmm:c-deltas} below states that all these conditional gains can be calculated explicitly by using the prediction law of the fractional Brownian motion.

\begin{lmm}[Conditional Gains]\label{lmm:c-deltas}
\begin{eqnarray*}
\Delta\hat S_{t_{i+1}}(t_i) &=&
S_{t_i}\left(\int_{-\infty}^\infty \e^{\mu\Delta t_{i+1}+\sigma\Delta \hat B_{t_{i+1}}(t_i)+\sigma\sqrt{\hat r(t_{i+1}|t_i)}z} \phi(z)\, \d z - 1\right), \\
\Delta\hat V^\pi_{t_{i+1}}(t_i)
&=&
\int_{-\infty}^\infty
f\left(S_{t_i}\e^{\mu\Delta t_{i+1}+\sigma\Delta \hat B_{t_{i+1}}(t_i)+\sigma\sqrt{\hat r(t_{i+1}|t_i)}z}\right) \phi(z)\, \d z - f(S_{t_i}), \\
\Delta\hat V^{\pi^N,k}_{t_{i+1}}(t_i)
&=&
\pi_{t_i}^N\Delta\hat S_{t_{i+1}}(t_i) - k S_{t_i}|\Delta \pi^N_{t_i}|.
\end{eqnarray*}
\end{lmm}

\begin{proof}
Let $g\colon\R\to\R$ be such that
$
\E\left[\left|g(B_{t_{i+1}})\right|\right] < \infty.
$
Then, by Lemma \ref{lmm:cond-fbm},
$$
\E\left[g(B_{t_{i+1}})\,|\, \F_{t_i}\right]
=
\int_{-\infty}^\infty g\left(\hat B_{t_{i+1}}(t_i)+\sqrt{\hat r(t_{i+1}|t_i)}z\right) \phi(z)\, \d z.
$$

Consider $\Delta\hat S_{t_{i+1}}(t_i)$.  By choosing
$$
g(x) = S_0\e^{\mu t + \sigma x},
$$
we obtain
\begin{eqnarray*}
\hat S_{t_{i+1}}(t_i)
&=&
\E\left[ S_{t_{i+1}}\,\big|\, \F_{t_i}\right] \\
&=&
\E\left[g\left(B_{t_{i+1}}\right)\,\big|\, \F_{t_i}\right] \\
&=&
\int_{-\infty}^\infty S_0\e^{\mu t_{i+1}+\sigma\left(\hat B_{t_{i+1}}(t_i)+\sqrt{\hat r(t_{i+1}|t_i)}z\right)} \phi(z)\, \d z.
\end{eqnarray*}
The formula for $\Delta\hat S_{t_{i+1}}(t_i)$ follows from this.

Consider then $\Delta V^\pi_{t_{i+1}}(t_i)$.
By choosing
$$
g(x) = f\left(S_0\e^{\mu t + \sigma x}\right)
$$
we obtain
\begin{eqnarray*}
\hat V^\pi_{t_{i+1}}(t_i)
&=&
\E\left[ V^\pi_{t_{i+1}}\, |\, \F_{t_i} \right] \\
&=&
\E\left[ f(S_{t_{i+1}})\, |\, \F_{t_i} \right] \\
&=&
\E\left[ g(B_{t_{i+1}})\, |\, \F_{t_i} \right] \\
&=&
\int_{-\infty}^\infty g\left(\hat B_{t_{i+1}}(t_i)+\sqrt{\hat r(t_{i+1}|t_i)}z\right) \phi(z)\, \d z\\
&=&
\int_{-\infty}^\infty
f\left(S_0\e^{\mu t_{i+1} + \sigma\left(\hat B_{t_{i+1}}(t_i)+\sqrt{\hat r(t_{i+1}|t_i)}z\right)}\right) \phi(z)\, \d z.
\end{eqnarray*}
The formula for $\Delta\hat V^{\pi}_{t_{i+1}}(t_i)$ follows from this.

Finally, we calculate
\begin{eqnarray*}
\hat V^{\pi^N,k}_{t_{i+1}}(t_i)
&=&
\E\left[ V^{\pi^N,k}_{t_{i+1}}\, \big| \, \F_{t_i}\right] \\
&=&
V^{\pi^N,k}_{t_i} +
\E\left[\int_{t_{i}}^{t_{i+1}} \pi^N_u \, \d S_u - \int_{t_i}^{t_{i+1}} k S_{u}|\d\pi^N_{u}|\,\Big|\, \F_{t_i}\right] \\
&=&
V^{\pi^N,k}_{t_i} + \pi_{t_i}^N\left(\E\left[S_{t_{i+1}}\big| \F_{t_i}\right] - S_{t_i}\right) - k S_{t_i}|\Delta \pi^N_{t_i}| \\
&=&
V^{\pi^N,k}_{t_i} + \pi_{t_i}^N\Delta\hat S_{t_{i+1}}(t_i) - k S_{t_i}|\Delta \pi^N_{t_i}|.
\end{eqnarray*}
The formula for $\Delta\hat V^{\pi^N,k}_{t_{i+1}}(t_i)$ follows from this.
\end{proof}

Now we are ready to state and prove our main result.  We note that, in principle, our result is general: it is true in any pricing model where the option $f(S_T)$ can be replicated. In practice, our result is specific to the fractional Black--Scholes model via Lemma \ref{lmm:c-deltas}.

\begin{thm}[Conditional-Mean Hedging Strategy]\label{thm:dhedging}
The conditional mean hedge of the European vanilla type option with convex or concave positive payoff function $f$ with proportional transaction costs $k$ is given by the recursive equation
\begin{equation}\label{eq:dhedging}
\pi^N_{t_i}
=
\frac{\Delta\hat V^\pi_{t_{i+1}}(t_i) + (V^\pi_{t_i}-V^{\pi^N,k}_{t_i}) + k S_{t_i}|\Delta\pi^N_{t_i}|}{\Delta\hat S_{t_{i+1}}(t_i)},
\end{equation}
where $V^{\pi^N,k}_{t_i}$ is determined by \eqref{eq:kappa}.
\end{thm}

\begin{proof}
Let us first consider the left hand side of \eqref{eq:cm-hedge}. We have
\begin{eqnarray*}
\E\left[V^{\pi^N,k}_{t_{i+1}}\,\big|\, \F_{t_i}\right]
&=&
\E\left[V^{\pi^N,k}_{t_{i}} + \int_{t_i}^{t_{i+1}}\pi^N_u\, \d S_u
-k \int_{t_i}^{t_{i+1}} S_u |\d\pi_u^N|\,\Big|\, \F_{t_i}\right] \\
&=&
V^{\pi^N,k}_{t_i}
+ \pi^N_{t_i} \E\left[ S_{t_{i+1}}(t_i) - S_{t_i}\,\big|\, \F_{t_i}\right] - kS_{t_i}|\Delta\pi^N{t_i}| \\
&=&
V^{\pi^N,k}_{t_i}
+ \pi^N_{t_i} \Delta\hat S_{t_{i+1}}(t_i) - kS_{t_i}|\Delta\pi^N_{t_i}|.
\end{eqnarray*}
For the right-hand-side of \eqref{eq:cm-hedge}, we simply write
\begin{eqnarray*}
\E\left[V^{\pi}_{t_{i+1}}\,\big|\, \F_{t_i}\right]
&=&
\Delta\hat V^\pi_{t_{i+1}}(t_i) + V^\pi_{t_i}.
\end{eqnarray*}
Equating the sides we obtain \eqref{eq:dhedging} after a little bit of simple algebra.
\end{proof}

\begin{rem}
Taking the expected gains $\Delta\hat S_{t_{i+1}}(t_i)$ to be the num\'eraire, one recognizes three parts in the hedging formula \eqref{eq:dhedging}.   First, one invests on the expected gains in the time-value of the option. This ``conditional-mean Delta-hedging'' is intuitively the most obvious part.  Indeed, a na\"ive approach to conditional-mean hedging would only give this part, forgetting to correct for the tracking-errors already made, which is the second part in \eqref{eq:dhedging}.  The third part in \eqref{eq:dhedging} is obviously due to the transaction costs.
\end{rem}

\begin{rem}
The equation \eqref{eq:dhedging} for the strategy of the conditional-mean hedging is recursive: in addition to the filtration $\F_{t_i}$, the position $\pi^N_{t_{i-1}}$ is needed to determine the position $\pi^N_{t_i}$. Consequently, to determine the conditional-meand hedging strategy by using \eqref{eq:dhedging}, the initial position $\pi^N_0$ must be fixed.  The initial position is, however, not uniquely defined.  Indeed, let $\beta^N_0$ be the position in the riskless asset. Then the conditional-mean criterion \eqref{eq:cm-hedge} only requires that
$$
\beta^N_0 + \pi^N_0\E[S_{t_1}] -kS_0|\pi^N_0|
=
\E[f(S_{t_1})].
$$
There are of course infinite number of pairs $(\beta^N_0,\pi^N_0)$ solving this equation.  A natural way to fix the initial position $(\beta^N_0,\pi^N_0)$ for the investor interested in conditional-mean hedging would be the one with minimal cost.  If short-selling is allowed, the investor is then faced with the minimization problem
$$
\min_{\pi^N_0 \in \R} v(\pi^N_0),
$$
where the initial wealth $v$ is the piecewise linear function
\begin{eqnarray*}
v(\pi^N_0)
&=&
\beta^N_0 + \pi^N_0 S_0 \\
&=&
\left\{\begin{array}{rl}
\E[f(S_{t_1})] - \left(\Delta \hat S_{t_1}(0)-kS_0\right)\pi^N_0,
& \mbox{ if }\quad \pi^N_0\ge 0,\\
\E[f(S_{t_1})] - \left(\Delta \hat S_{t_1}(0)+kS_0\right)\pi^N_0,
& \mbox{ if }\quad \pi^N_0 < 0.
\end{array}\right.
\end{eqnarray*}
Clearly, the minimal solution $\pi^N_0$ is independent of $\E[f(S_{t_1})]$, and, consequently, of the option to be replicated.   Also, the minimization problem is bounded if and only if
\begin{eqnarray*}
k \ge \left|\frac{\Delta\hat S_{t_1}(0)}{S_0}\right|,
\end{eqnarray*}
i.e. the proportional transaction costs are bigger than the expected return on $[0,t_1]$ of the stock.  In this case, the minimal cost conditional mean-hedging strategy starts by putting all the wealth in the riskless asset.
\end{rem}

We end this note by applying Theorem \ref{thm:dhedging} to European call options.

\begin{cor}[European Call Option]\label{cor:call}
Denote
\begin{eqnarray*}
\hat d^+_{t_{i+1}}(t_i)
&=&
\frac{\ln\frac{S_{t_i}}{K} - \mu\Delta t_{i+1} - \sigma\Delta\hat B_{t_{i+1}}(t_i)}{\sigma \sqrt{\hat r(t_{i+1}|t_i)}} - \sigma\sqrt{\hat r(t_{i+1}|t_i)},  \\
\hat d^-_{t_{i+1}}(t_i)
&=&
\frac{\ln\frac{S_{t_i}}{K} - \mu\Delta t_{i+1} - \sigma\Delta\hat B_{t_{i+1}}(t_i)}{\sigma \sqrt{\hat r(t_{i+1}|t_i)}}, \\
\hat X_{t_{i+1}}(t_i) &=& \mu\Delta t_{i+1} + \sigma\Delta\hat B_{t_{i+1}}(t_i)+\frac12\sigma^2\hat r(t_{i+1}|t_i),
\end{eqnarray*}
and let $\Phi$ be the cumulative distribution function of the standard normal law.
Then the conditional-mean hedging strategy for the European call option with strike-price $K$ is given by
\begin{equation}\label{eq:call-cmh}
\pi^N_{t_i} =
\frac{S_{t_i}\e^{\hat X_{t_{i+1}}(t_i)}\Phi(\hat d^+_{t_{i+1}}(t_i)) - K\Phi(\hat d^-_{t_{i+1}}(t_i))-V^{\pi^N,k}_{t_i} + k S_{t_i}|\Delta\pi^N_{t_i}|}{\Delta\hat S_{t_{i+1}}(t_i)}.
\end{equation}
\end{cor}

\begin{proof}
First we note that
\begin{eqnarray*}
\hat V^\mathrm{call}_{t_{i+1}}(t_i)
&=&
\int_{-\infty}^\infty
\left(S_{t_i}\e^{\mu \Delta t_{i+1} + \sigma\Delta\hat B_{t_{i+1}}(t_i)+\sigma\sqrt{\hat r(t_{i+1}|t_i)}z}-K\right)^+ \phi(z)\d z \\
&=&
S_{t_i}\e^{\hat X_{t_{i+1}}(t_i)}\Phi\left(\hat d^+_{t_{i+1}}(t_i)\right) - K\Phi\left(\hat d^-_{t_{i+1}}(t_i)\right).
\end{eqnarray*}
Next we note that
$$
V^\mathrm{call}_{t_i}
= (S_{t_i}-K)^+.
$$
So,
$$
\Delta\hat V^\mathrm{call}_{t_{i+1}}(t_i)
=
S_{t_i}\e^{\hat X_{t_{i+1}}(t_i)}\Phi(\hat d^+_{t_{i+1}}(t_i)) - K\Phi(\hat d^-_{t_{i+1}}(t_i)) - (S_{t_i}-K)^+,
$$
and \eqref{eq:call-cmh} follows from this.
\end{proof}


\bibliographystyle{siam}
\bibliography{pipliateekki}
\end{document}